\documentclass[a4paper,UKenglish,cleveref, autoref, thm-restate,authorcolumns]{lipics-v2021}

\hideLIPIcs  %uncomment to remove references to LIPIcs series (logo, DOI, ...), e.g. when preparing a pre-final version to be uploaded to arXiv or another public repository

\nolinenumbers %uncomment to disable line numbering

\title{Aspects of Coherence in Dependence Logic}

\author{Timon Barlag}{Leibniz Universität Hannover, Germany}{barlag@thi.uni-hannover.de}{https://orcid.org/0000-0001-6139-5219}{}

\author{Nicolas Fröhlich}{Leibniz Universität Hannover, Germany}{nicolas.froehlich@thi.uni-hannover.de}{https://orcid.org/0009-0003-5413-1823}{The author appreciates funding by the German Research Agency (DFG) under the grant ME2479/3-1 and project id 511769688}

\author{Miika Hannula}{University of Tartu, Estonia\\University of Helsinki, Finland}{miika.hannula@ut.ee}{https://orcid.org/0000-0002-9637-6664}{The author appreciates funding by the ERC grant 101020762}

\author{Phokion G.\ Kolaitis}{ University of California Santa Cruz,  USA}{kolaitis@ucsc.edu}{https://orcid.org/0000-0002-8407-8563}{The author acknowledges support from the University of Helsinki during a two-month visit in the spring of 2026.}

\author{Juha Kontinen}{University of Helsinki, Finland}{juha.kontinen@helsinki.fi}{https://orcid.org/0000-0003-0115-5154}{}

\author{Arne Meier}{Leibniz Universität Hannover, Germany}{meier@thi.uni-hannover.de}{https://orcid.org/0000-0002-8061-5376}{The author appreciates funding by the German Research Agency (DFG) under the grant ME2479/3-1 and project id 511769688}

\author{Jouko Väänänen}{University of Helsinki, Finland}{jouko.vaananen@helsinki.fi}{https://orcid.org/0000-0003-4356-7974}{This author has received funding from the European Research Council (ERC) under the European Union’s Horizon 2020 research and innovation programme (grant agreement No 101020762), and from the Research Council of Finland (grant No. 368671).}

\authorrunning{T. Barlag et al.} %TODO mandatory. First: Use abbreviated first/middle names. Second (only in severe cases): Use first author plus 'et al.'


 %TODO mandatory; please add comma-separated list of keywords

\usepackage{booktabs}

\newcommand{\itemstyle}[1]{{\color{lipicsGray}\normalfont\bfseries\sffamily{#1}}}

\usepackage{xspace}

\usepackage{mathtools,amssymb,amsmath}
\usepackage{hyperref}
\usepackage{cleveref}
\newcommand{\calA}{\mathcal{A}}
\newcommand{\calB}{\mathcal{B}}

\newcommand{\VAR}{\mathrm{VAR}}
\newcommand{\dep}{\mathrm{dep}}
\newcommand{\DL}{\mathcal{D}}

\newcommand{\coRE}{\mathsf{coRE}}
\newcommand{\RE}{\mathsf{RE}}
\newcommand{\NP}{\mathsf{NP}}

\newcommand{\LOGSPACE}{\mathsf{LOGSPACE}}
\newcommand{\NLOGSPACE}{\mathsf{NLOGSPACE}}

\newcommand{\NEXPTIME}{\mathsf{NEXPTIME}}
\newcommand{\coNEXPTIME}{\mathsf{coNEXPTIME}}

\newcommand{\PiE}[1]{\Pi^\mathsf{EXP}_{#1}}

\newcommand{\UDL}{\forall\DL}
\newcommand{\FO}{\mathrm{FO}}
\newcommand{\DLQF}{\DL\mathrm {QF}}
\newcommand{\oN}{{\mathbb{N}}}
\newcommand{\SNP}{\mathrm {SNP}}
\newcommand{\PDL}{\mathcal{PD\hspace*{-0.1em}L}}
\newcommand{\kCOHERENCE}[1][$k$]{#1-{\sc Coherence}\xspace}
\newcommand{\COHERENCE}{{\sc Coherence}\xspace}
\newcommand{\tmodels}[3]{{#1} \models_{#2} {#3}}
\newcommand{\tnotmodels}[3]{{#1} \not\models_{#2} {#3}}

\newcommand{\leqbtt}{\preceq^p_{b\text-tt}}

\newcommand{\Rel}{\operatorname{Rel}}

\newcommand{\commentout}[1]{{}}
\usepackage[most]{tcolorbox}
\newtcolorbox{mytodo}[2][]{
  arc=5mm,
  lower separated=false,
  enhanced,
  breakable,
  title={#2}, % Set the title explicitly
  #1 % Allow additional optional parameters
}

\begin{document}

\maketitle
\begin{abstract}
%%%%
Dependence logic extends first-order logic with dependence atoms asserting that the value of a variable is determined by the values of certain other variables.
%, thus  dependence atoms express functional dependencies in database theory.
The semantics of dependence logic has a second-order character and involves  sets of assignments, called teams, instead of individual assignments as in the classical Tarski semantics. Since the model-checking problem  is known to be $\NP$-complete even for quantifier-free dependence logic (DQF) formulas, researchers have pursued conditions on formulas that make this problem tractable.
In 2010, Jarmo Kontinen introduced the notion of $k$-coherence for dependence logic formulas, where $k$ is a positive integer. This notion asserts  that if the formula is satisfied in a structure by all $k$-element subteams of a given team, then the given team itself satisfies the formula.  It has been proved that $k$-coherent DQF-formulas have a tame model-checking problem, because such formulas admit a first-order rewriting.
In this paper, we investigate the structural
and algorithmic aspects of coherence. We  show that if a DQF-formula is first-order rewritable, then it is $k$-coherent for some positive integer $k$. Thus, for DQF-formulas, coherence is equivalent to first-order rewritability. Furthermore, we show that an analogous result holds for universally quantified dependence logic formulas under  a stronger notion of coherence. After this, we focus on the complexity of deciding if a given dependence logic formula is $k$-coherent. We establish that this decision problem is highly undecidable for arbitrary dependence logic formulas, while for DQF-formulas this problem is co-recursively enumerable. Furthermore, we pinpoint the computational complexity of the coherence problem for propositional dependence logic formulas by showing that this problem is complete for the second level of the exponential hierarchy.

\end{abstract}

\section{Introduction}
Team semantics is a  framework for studying concepts and phenomena that inherently involve collections of objects rather than individual ones.
%\phokion{Is it accurate to say that  team semantics is ``a logical framework''? Do not we need both syntax and semantics to talk about a "logical framework"? Also, "logical framework" is a technical term in type theory. Let's discuss.}\jouko{I changed it to just ``framework".}
 Functional dependence in database theory and conditional independence of random variables in statistics are examples of such concepts.
Team semantics was introduced by Hodges~\cite{hodges1997compositional} to give  semantics to independence friendly logic and by Väänänen~\cite{Vaananen-book} to give semantics to dependence logic.
In these logics,  formulas are evaluated over sets of assignments, called \emph{teams}, rather than over single assignments, as in the classical Tarski semantics. Here, our  focus is on dependence logic.
The syntax of  dependence logic $\DL$ extends that of first-order logic $\FO$ with new atomic formulas $\dep(\vec x,y)$, called \emph{dependence atoms}; these atoms
 assert that the values of the variables $\vec x$  determine the value of the variable~$y$, hence they   express  functional dependencies in database theory.
The notion of satisfaction $\calA \models_T \varphi$, where $\varphi$ is a $\DL$-formula, involves both a structure $\calA$ and a team $T$. While the syntax of $\DL$ resembles that of $\FO$, its semantics is second-order. For example, $\calA \models_T \varphi_1 \lor \varphi_2$ if there are two subteams $T_1$ and $T_2$ of $T$ such that $T=T_1\cup T_2$,
$\calA  \models_{T_1}  \varphi_1$, and $\calA \models_{T_2} \varphi_2$.

In the past decade, the expressive power and computational complexity of dependence logic and related team-based logics have been studied extensively (see, e.g., \cite{DurandKRV22,KontinenMM22,FKM2026,DurandKV24,10.1145/3771721}).
Furthermore, connections to several different areas have been established, including database theory~\cite{HannulaKL16,HannulaK16,Engstrom12}, meta-finite model theory~\cite{HannulaKBV20}, and inquisitive logic~\cite{CiardelliIY20}.
The framework of team semantics has also been generalized in various directions, including  multiteams~\cite{DurandHKMV18}, probability distributions~\cite{HKMV18}, temporal logics~\cite{DBLP:conf/lics/GutsfeldMOV22,DBLP:conf/mfcs/KrebsMV018}, and semiring annotations~\cite{BarlagHKPV23}.

The model-checking problem for a $\DL$-formula $\varphi$ asks: given a structure $\calA$
and a team $T$, does $\calA \models_T \varphi$?
Because of the second-order character of the semantics of $\DL$, the model-checking problem for $\DL$-formulas is
 drastically different than that for $\FO$-formulas. Specifically, while the model-checking problem for $\FO$-formulas is always in $\LOGSPACE$, the model-checking problem for $\DL$-formulas can be $\NP$-complete.
In fact, as shown by Jarmo Kontinen~\cite{JaKontinen13}, it is $\NP$-complete even for the quantifier-free formula $\DL$-formula $\dep(x,y) \vee \dep(u,v) \vee \dep(u,v)$.

 This state of affairs has motivated the pursuit of syntactic patterns and model-theoretic conditions on $\DL$-formulas, so that the model-checking problem becomes tractable and, ideally,  no more difficult than the model-checking of $\FO$-formulas.
 As regards syntactic patterns, it was shown by Jarmo Kontinen~\cite{JaKontinen13} that the model-checking problem for disjunctions of two dependence atoms is always in $\NLOGSPACE$.
 More recently, a complete classification of the complexity of the model-checking problem for disjunctions of two dependence atoms was established by Fröhlich~et~al.~\cite{FKM2026}.

As regards model-theoretic conditions, Jarmo Kontinen~\cite{JarmoThesis,JaKontinen13} introduced and investigated the notion
of $k$-coherence, where $k$ is a positive integer. A $\DL$-formula $\varphi$ is $k$-coherent if
for every structure  $\calA$ and every team $T$ such that $\calA\models_T' \varphi$ holds for every $k$-element subteam $T'$ of $T$, it holds that $\calA \models \varphi$. Thus, $k$-coherence is a ``small-model'' property, since it is equivalent to asserting that
if $\calA\not\models_T \varphi$, then there is a $k$-element subteam $T'$ of $T$ such that $\calA\not \models_{T'}\varphi$.
Jarmo Kontinen showed~\cite{JarmoThesis,JaKontinen13} that if $\varphi$ is a  $k$-coherent quantifier-free $\DL$-formula, then the model-checking problem for $\varphi$ is well-behaved, in the sense that it admits a $\FO$-\emph{rewriting}, i.e., there is a $\FO$-formula $\varphi^*(R)$ with an extra relation symbol $R$ such that for every structure $\calA$ and every team $T$, we have that
$\calA \models_T\varphi$ if and only if
$(\cal A, \Rel(T)) \models \varphi^*(R)$, where $\Rel(T)$ is the relation represented by the team $T$; furthermore, $\varphi^*$ can be taken to be a universal $\FO$-formula. Jarmo Kontinen also showed~\cite{JarmoThesis} that this $\FO$-rewritability result extends to \emph{strongly} $k$-coherent universal $\DL$-formulas.

%Interestingly, $k$-coherence also plays an important role in inquisitive logic, where it can be used to characterize fragments of inquisitive first-order logic with favourable meta-theoretical properties \cite{CIARDELLI2022103155}.

%It is worth noting that, although inquisitive first-order logic satisfies a form of downward closure, it differs substantially from dependence logic, both in its semantics, expressivity and in the choice of primitive logical operators (for details, see the work of Kontinen and Viardelli~\cite{KC26}).

%\subparagraph*{Summary of Results}
In this paper, we embark on a systematic study of structural and algorithmic aspects of coherence and its variants for dependence logic. As discussed in the preceding paragraph, coherence of a quantifier-free $\DL$-formula is a sufficient condition for $\FO$-rewritability. Fröhlich~et~al.~\cite{FKM2026} left as an open problem whether the converse is true, i.e., does $\FO$-rewritability imply $k$-coherence for some $k\geq 1$? Using
 the Łoś-Tarski preservation theorem of first-order logic~\cite{MR72071,MR66303}, we show that this is indeed true.  Therefore, for a quantifier-free $\DL$-formula $\varphi$, the following statements are equivalent: (i) $\varphi$ is $k$-coherent for some $k\geq 1$; (ii) $\varphi$ is
 rewritable to a universal $\FO$-formula; (iii) $\varphi$ is $\FO$-rewritable.
 We then consider universal $\DL$-formulas and show that an entirely analogous rewritability result holds for stongly coherent universal $\DL$-formulas.
 While for quantifier-free formulas $\DL$-formulas strong coherence and coherence coincide, this no longer holds true for universal $\DL$-formulas; to this effect,
  we give an example of a universal $\DL$-formula that is $1$-coherent, but not $\FO$-rewritable.

%We then turn to study the possibility of generalizing the aforementioned result for larger classes of dependence logic formulas.
%Recall {\color{red}M: Or, ``It is known''} that all formulas of dependence logic can be transformed to the form $\forall \vec{x}\exists \vec{y}\theta$, where $\theta$ is a quantifier-free formula.
%Here, we show that the results from before can be extended to universal dependence logic formulas by means of the notion of strong coherence (which implies ordinary coherence).
%Eventually, we generalize $k$-coherence to $\lambda$-coherence, where $\lambda$ is an infinite cardinal, and show that universal formulas are $\omega$-coherent, whereas existential quantification enables the construction of formulas that are not $\lambda$-coherent for any infinite cardinal $\lambda$.
%Interestingly, a corresponding phenomenon regarding existential formulas in inquisitive first-order logic has also been observed by Ciardelli and Grilletti~\cite{CIARDELLI2022103155}.

After completing the investigation of  the structural aspects of coherence, we turn attention to its algorithmic aspects. Specifically, we consider the problem of deciding coherence, that is, given a $\DL$-formula $\varphi$, is $\varphi$ $k$-coherent? Our first result shows that, for every fixed $k\geq 2$, deciding $k$-coherence for arbitrary $\DL$-formulas is a problem of non-arithmetic complexity, which implies that it is  a highly undecidable problem. For quantifier-free $\DL$-formulas, however, we show that deciding $k$-coherence is a co-recursively enumerable problem, even when $k$ is part of the input. For this result, we use  the compactness theorem for first-order logic to show that, for quantifier-free $\DL$-formulas, $k$-coherence coincides with $k$-coherence in the finite, i.e., it suffices to consider finite structures and finite teams.
The decidability of coherence for quantifier-free  $\DL$-formulas remains an open problem.
Finally, we explore coherence for propositional dependence logic, a variant of
dependence logic that has been studied earlier~\cite{DBLP:journals/iandc/Virtema17,YangV17,Hannula19}.
Our main complexity-theoretic result is that, for every fixed $k\geq 2$, deciding
$k$-coherence for propositional dependence logic formulas is a $\PiE{2}$-complete problem via polynomial-time  bounded-truth-table reductions, where  $\PiE{2}$ is
 the second level of the exponential hierarchy.

% \phokion{Note that I did not mention the $\coRE$-ness of coherence for universal $\DL$-formulas because coherence is not enough for $\FO$-rewritability, so this result may feel unmotivated.  Let's discuss.}

\subparagraph*{Structure of the paper}
Section~\ref{sec:preliminaries}  provides background to dependence logic $\DL$.
Section~\ref{sec:coherence and rewritability} studies coherence and $\FO$-rewritability for quantifier-free $\DL$-formulas.
Section~\ref{sec:strong coherence} investigates coherence for quantified $\DL$-formulas.
Section~\ref{sec:complexity} examines the algorithmic aspects of coherence.

\subparagraph*{Related work}
We note that $k$-coherence also plays an important role in inquisitive logic, where it is used to characterize fragments of inquisitive first-order logic with favourable model-theoretical properties~\cite{CIARDELLI2022103155}.
% The notion of $k$-coherence has been also used in temporal team semantics to identify decidable fragments of temporal team logics. Furthermore, deciding $k$-coherence has been shown to be undecidable also in this context \cite{VirtemaHFK021}.
The notion of $k$-coherence has also been used in modal dependence logic~\cite{DBLP:conf/aiml/HellaLSV14} and in temporal team logics~\cite{VirtemaHFK021}.

\section{Preliminaries}\label{sec:preliminaries}

%\arne{define vocabularies, domain of $\calA$, $\VAR$, ...}
%\arne{Define only our relevant stuff and not the full logic?}
We make use of standard notions in complexity theory~\cite{DBLP:books/daglib/0018514}.
A \emph{vocabulary} is a set of relation and constant symbols.
%tuple $\tau = (R_1, \dots, R_m)$ of relation symbols.
A \emph{$\tau$-structure}
 $\calA$ % = (A, R_1^{\calA}, \dots R_m^{\calA})$,
 comprises of a universe $A$ and of an $r$-ary relation $R^{\calA}$ and  element $c^{\calA}$ of $A$  %is a  relation (element) on $A$,
 for each $r$-ary relation symbol $R$ and constant symbol $c$ in $\tau$. %$1 \leq i \leq m$.

Let $\VAR$ be a countably infinite set of variables.
A \emph{dependence atom} is an expression of the form $\dep(\vec y,z)$, where $\vec y = y_1,\ldots,y_n$ and $y_1,\ldots,y_n,z \in \VAR$. % are distinct variables.
% Let $\tau$ be a relational vocabulary.
The formulas of dependence logic $\DL$ over $\tau$ are defined by the following grammar:
\[
	\psi \coloneqq
    R_i(\vec x_i)\mid
    \lnot R_i(\vec x_i)\mid
    x_k = x_\ell \mid
     \lnot x_k = x_\ell \mid
    \dep(\vec y, z) \mid
    \psi \land \psi \mid
    \psi \lor \psi \mid
    \forall x \psi \mid
    \exists x \psi,
\]
where $x_k,x_\ell,\vec x_i$ are (not necessarily distinct) constants or variables in $\VAR$, $R_i$ is a relation symbol in $\tau$, and $y_1,\ldots,y_n,z$ are distinct variables in $\VAR$.
We write $\DL(\tau)$ to denote the collection of all  formulas of dependence logic $\DL$ over $\tau$. When $\tau=\emptyset$, we write just $\DL$.

\begin{definition}
    Let $V$ be a set of variables (i.e., $V\subseteq \VAR$) and let $A$ be a set.
    \begin{itemize}
    \item An \emph{assignment with domain $V$ and range $A$}
    is a mapping $s\colon V \to A$.

    \item A \emph{team with domain $V$ and range $A$} is a set $T$ of assignment with domain $V$ and range $A$.
    We will use the notation $\Rel(T) \coloneqq \{\,(s(x_1), \dots, s(x_n)) \mid s \in T\,\} \subseteq A^n$, for $x_1, \dots, x_n \in V$, when $T$ is considered as a relation instead of a set of assignments.
    % We will write $\Dom(T)$ to denote the domain $V$ of the team $T$.
    % Similar $\Range(T)$ denotes the range $A$ of the team $T$.
    % \item If $T$ is a team with domain $V$ and range $A$ and $B \subseteq A$, then the \emph{restriction of $T$ on $B$} is the team $T \upharpoonright B = \{\,s \mid \bigcup_{x \in V} s(x) \subseteq B\,\}$.
    % Likewise if $U \subseteq V$, then the \emph{restriction of $T$ on $U$} is the team $T\upharpoonright U = \{\,s\upharpoonright U \mid s\in T\,\}$, where $s\upharpoonright U$ is the restriction of the assignment $s$ to $U$.
    \end{itemize}
\end{definition}

%\juha{we should not assume that s is defined on all of VAR}
%\nicolas{changed it}

The next definition describes some operations on teams that will be used to give semantics to the universal and existential quantifiers of dependence logic.

\begin{definition}
    Let $V$ be a set of variables (i.e., $V\subseteq \VAR$) and let $A$ be a set.
    \begin{itemize}
    \item If $s\colon V \to A$ is an assignment with domain $V$ and range $A$, $x$ is a variable, and $a$ is an element of $A$, then $s^x_a\colon V \cup \{x\} \to A$ is the assignment such that $s^x_a(x) = a$ and $s^x_a(y) = s(y)$, for every variable $y\not = x$.

    \item If $T$ is a team with domain $V$ and range $A$, $f\colon T \to \mathcal{P}(A) \setminus \{\emptyset\}$ is a function, and $x$ is a variable, then the \emph{supplement team} $T^x_f$ is the team $ \{\,s^x_{a} \mid s \in T\text{ and }a\in f(s)\,\}$, with extended domain $V \cup \{x\}$.

    \item If $T$ is a team with domain $V$ and range $A$, and $x$ is a variable, then the \emph{duplicate team} $T^x_A$ is the team $ \{\,s^x_a \mid s \in T \text{ and } a \in A \,\}$, with extended domain $V \cup \{x\}$.
    \end{itemize}
\end{definition}

We are now ready to give the definition of the semantics of dependence logic formulas.

% \juha{We are now usign the so-called strict interpretation for the existential quantifier. It is of course equivalent but might cause confusion.}
% \arne{Mabye we should then add a sentence for clarification in this direction to prevent obvious reviewer comments?}
% \miika{Currently Defs 1-4 in a row. Would it look better to take e.g. Defs 1,2,4 out of definition environments?}
% \arne{With 1+2 out of def env i am fine. I think 4 should stay inside}
\begin{definition}\label{def:semantics}
	Let $\tau
    % =(R_1,\ldots,R_m)
    $ be a vocabulary and let $\calA
    % =(A,R_1^\calA,\ldots,R_m^\calA)
    $ be  a $\tau$-structure.
    The satisfaction relation $\tmodels{\calA}{T} \varphi$, where $T$ is a team with domain $\VAR$ and range $A$, and $\varphi$ is a formula of dependence logic $\DL$, is defined by induction on the construction of $\varphi$ and simultaneously for all teams with domain $\VAR$ and range $A$ as follows, where $\vec y = y_1,\ldots,y_n$, and $x_1,\dots,x_r,x_k,x_\ell\in\VAR$:
    \begin{tabbing}
$\mathcal A\models_TR(x_1,\dots,x_r)$\quad\=\;if\;\=\kill
$\tmodels{\calA}{T}{R(x_1,\dots,x_r)}$ \> if \> for all $s\in T$, we have $(s(x_1),\ldots,s(x_{r})) \in R^{\calA}$;\\
$\tmodels{\calA}{T}{\neg R(x_1,\dots,x_r)}$ \> if \> for all $s\in T$, we have $(s(x_1),\ldots,s(x_{r})) \not\in R^{\calA}$;\\
$\tmodels{\calA}{T}{x_k= x_\ell}$ \> if \> for all $s\in T$, we have $s(x_k) = s(x_\ell)$;\\
$\tmodels{\calA}{T}{\neg x_k= x_\ell}$ \> if \> for all $s\in T$, we have $s(x_k) \neq s(x_\ell)$;\\
$\tmodels{\calA}{T}{\dep(\vec y,z)}$ \> if \> for all $s_1,s_2\in T$: $s_1(y_j)=s_2(y_j)$ for all $j$ implies $s_1(z)=s_2(z)$;\\
$\tmodels{\calA}{T}{\varphi_1\land \varphi_2}$ \> if \> $\tmodels{\calA}{T}{\varphi_1}$ and $\tmodels{\calA}{T}{\varphi_2}$;\\
$\tmodels{\calA}{T}{\varphi_1\lor \varphi_2}$ \> if \> there are teams $T_1$, $T_2$ with $T= T_1\cup T_2$,  $\tmodels{\calA}{T_1}{\varphi_1}$, $\tmodels{\calA}{T_2}{\varphi_2}$;\\
$\tmodels{\calA}{T}{\exists x\varphi}$ \> if \> there is a function $f\colon T\to A$ such that $\tmodels{\calA}{T^x_f}{\varphi}$,\\
\>\> where $T^x_f$ is the supplement team associated with $T$, $x$, and $f$;\\
$\tmodels{\calA}{T}{\forall x\varphi}$ \> if \> $\tmodels{\calA}{T^x_A}{\varphi}$, where $T^x_A$ is the duplicate team associated with $T$ and $x$.
\end{tabbing}
\end{definition}
%\arne{mention locality and that we are only interested in very specific formulas, so we can get rid of $\mathcal A$}

%Let $\mathcal C\subseteq\DL$ be a subset of the formulas from $\DL$.
%Then we will call formulas from $\mathcal C$ also $\mathcal C$-formulas.

%If $\Sigma$ is a set of $\DL$-formulas and $\varphi$ is a $\DL$-formula, then we say that $\varphi$ is \emph{implied by} $\Sigma$, $\Sigma\models\varphi$ in symbols, if for all teams $T$ we have that $T\models\Sigma$ implies $T\models\varphi$.

%\problemdef{$\IMP(\mathcal C)$ --- the implication problem for formulas of the class $\mathcal C$}{a set $\Sigma$ of $\mathcal C$-formulas, a $\mathcal C$-formula $\varphi$}{Is $\Sigma\models\varphi$ true}

In this article, we consider the following fragments of dependence logic.

\begin{definition} Let $\tau$ be a relational vocabulary.
    \begin{itemize}
        \item We write $\UDL(\tau)$ to denote the collection of $\DL$-formulas in which $\exists$ does not occur.
        % Equivalently, these are the formulas defined as follows:
        % \[
        % 	\psi \coloneqq
        %     R_i(\vec x_i)\mid
        %     \lnot R_i(\vec x_i)\mid
        %     x_k = x_\ell \mid
        %     \dep(\vec y, z) \mid
        %     \psi \land \psi \mid
        %     \psi \lor \psi \mid
        %     \forall x \psi.
        % \]
        \item We write $\DLQF(\tau)$ for the collection of all quantifier-free formulas of $\DL$ over $\tau$.
    \end{itemize}
    For the case of $\tau=\emptyset$, we write just $\UDL$, and $\DLQF$.
\end{definition}

The next result is a well-known normal form theorem for dependence logic.

\begin{theorem}[\cite{Vaananen-book}]
The following statements are true:
\begin{itemize}
    \item Every $\DL(\tau)$-formula is logically equivalent to a formula of the form
    \[
        \forall \vec{x}\exists \vec{y}\big(\bigwedge_i \dep(\vec{z}_i,y_i) \wedge \theta\big),
    \]
    where $\theta$ is a quantifier-free first-order formula over $\tau$.
    \item Every $\UDL(\tau)$-formula is logically equivalent to a formula of the form $\forall \vec{x} \theta$, where $\theta\in \DLQF(\tau)$.
\end{itemize}
\end{theorem}

\section{Coherence and rewritability}\label{sec:coherence and rewritability}
In this section, we formally introduce coherence and rewritability.
% We will show (in Theorem~\ref{thm:FO rewrite implies coherent}) that a $\DLQF(\tau)$-formula $\varphi$ is coherent if and only if it is first-order rewritable; furthermore, this happens precisely when $\varphi$ is rewritable to a universal first-order sentence.
% Afterwards, we study the concept of coherence in the finite and show (in Theorem~\ref{lem:finite coherence = coherence}) its equivalence to ``classical'' coherence for $\UDL(\tau)$ formulas.

\begin{definition}[Coherence, \cite{JaKontinen13}]\label{def:coherence}
    Suppose $\varphi(x_1, \dots, x_n)$ is a $\DL(\tau)$-formula and $k \in \oN$.
    We say that $\varphi$ is \emph{$k$-coherent} if for every (finite or infinite) $\tau$-structure $\calA$ and every team $T$ with domain $\{x_1,\ldots,x_n\}$ and range $A$, the following are equivalent:
    \begin{enumerate}
        \item $\tmodels{\calA}{T}{\varphi}$.
        \item For all $k$-element subteams $S\subseteq T$, it holds that $\tmodels{\calA}{S}{\varphi}$.
    \end{enumerate}
    We say that $\varphi$ is \emph{coherent} if there is a $k\geq 1$ such that $\varphi$ is $k$-coherent.
\end{definition}

%In what follows, we will write $\DLQF(\tau)$ to denote the collection of all quantifier-free formulas of dependence logic $\DL$ over a relational vocabulary $\tau$.
%When the vocabulary $\tau$ is empty, we will write just $\DLQF$.

% In this section we show that a $\DLQF(\tau)$-formula $\varphi$ is coherent  if and only if it is first-order rewritable; furthermore, this happens precisely when $\varphi$ is rewritable to a   universal first-order sentence.  We also show that for every $\DLQF(\tau)$-formula, the notion of coherence coincides with the notion of coherence in the finite.
An important observation is that the direction ``\itemstyle{1.} $\Longrightarrow$ \itemstyle{2.}'' is true for all formulas of $\DL(\tau)$ due to the \emph{downward closure} property of dependence logic (Väänänen~\cite[Prop.~3.10]{Vaananen-book}), i.\,e., if a team satisfies some $\DL(\tau)$-formula, then the same is true for \emph{all} its subteams.

We proceed with the precise definition of rewritability.

\begin{definition} \label{def:rewrite}
    Let $\mathcal L$ be a logic and, for every relational vocabulary $\tau$, let $\mathcal L(\tau)$ be the formulas in $\mathcal L$ over the vocabulary $\tau$.
    We say that a $\DL(\tau)$-formula $\varphi(x_1,\ldots,x_n)$ is \emph{rewritable to a formula in $\mathcal L$} or, simply, \emph{$\mathcal L$-rewritable} if there is a ${\mathcal L}(\tau \cup \{R\})$-sentence $\varphi^*$, where $R$ is an $n$-ary relation symbol, such that for every (finite or infinite) $\tau$-structure $\calA$ and every team $T$ with domain $\{x_1,\ldots,x_n\}$, we have that
    \[
        \calA \models_T \varphi \Longleftrightarrow ({\calA, \Rel(T)})\models \varphi^*.
    \]
\end{definition}

By definition, a \emph{strict $\Sigma_1^1$}-formula is an
existential second-formula $\Psi$ whose first-order part is a universal $\FO$-formula, i.e., $\Psi$ is
of the form $\exists {\vec S}\hspace{0.05em}\forall {\vec x}\hspace{0.05em}\theta$, where $\theta$ is a quantifier-free first-order formula.
We write $\SNP(\tau)$ to denote the collection of all strict $\Sigma_1^1$-formulas over a vocabulary $\tau$.
We also write $\SNP$ to denote the collection of all
strict $\Sigma_1^1$-formulas.

The following result follows by inspecting the proof of Theorem 6.2 in the textbook of Väänänen~\cite{Vaananen-book}.

\begin{proposition}%[Jarmo's thesis]
\label{prop:qf-strict}
Suppose that  $\varphi(x_1, \dots, x_n)$ is a $\UDL(\tau)$-formula. Then $\varphi$ is $\SNP$-rewritable.
%there is a $\SNP(\tau\cup\{R\})$-sentence $\varphi$, where $R$ is an $n$-ary relation symbol, such that for every $\tau$-structure $\calA$ and every team $T$ with domain $\{x_1,\ldots,x_n\}$, we have that
%$$\calA \models_T \varphi \Longleftrightarrow ({\calA, \Rel(T)})\models \varphi.$$
\end{proposition}

%\begin{definition} Let $\varphi(x_1, \dots, x_n)$ be a $\DLQF(\tau)$-formula. We say that
%$\varphi$ is \emph{first-order rewritable} if there is a $\FO(\tau \cup \{R\})$-sentence $\varphi^*$, where $R$ is an $n$-ary relation symbol, such that for every $\tau$-structure $\calA$ and every team $T$ with domain $\{x_1,\ldots,x_n\}$, we have that
%$$\calA \models_T \varphi \Longleftrightarrow ({\calA, \Rel(T)})\models \varphi^*.$$
%\end{definition}

\begin{theorem}[{\cite[Thm.~4.9]{JaKontinen13}}]\label{coherent-FO-reduction}
    Suppose that  $\varphi(x_1, \dots, x_n)$ is a coherent $\DLQF(\tau)$-formula. Then $\varphi$ is first-order rewritable. In fact, it is rewritable to a universal first-order sentence
    $\varphi^*$.
  %  Then there is a sentence $\varphi^* \in \FO(\tau \cup \{R\})$, where $R$ is $n$-ary, such that for all $\tau$-structures $\mathcal M$ and for all teams $T$ with domain $\{x_1, \dots, x_n\}$, we have that
   % $(\mathcal M, T) \models \varphi(x_1, \dots, x_n)$
   %if and only if
   %\Leftrightarrow (\mathcal M,
   % $\Rel(T) \models \varphi^*(R)$.
\end{theorem}
Note that the rewritability of a coherent $\DLQF$-formula to a universal first-order sentence follows from the proof of Theorem 4.9 in the work of Jarmo Kontinen~\cite{JaKontinen13}.
We will now prove the converse of Theorem \ref{coherent-FO-reduction} and begin with an auxiliary result, the quantifier-free analogue of which is presented by Jarmo Kontinen~\cite{JarmoThesis}.
% \timon{Maybe rather ``in~\cite{JarmoThesis}''?}

We now introduce the following notation.
If $T$ is a team with domain $V$ and range $A$ and $B \subseteq A$, then the \emph{restriction of $T$ on $B$} is the team $T \upharpoonright B = \{\,s \mid \bigcup_{x \in V} s(x) \subseteq B\,\}$.

% \timon{Move the ``likewise'' sentence to where it is needed}

\begin{lemma}\label{lem:closed under substructures}
    Let $\calA$ and $\calB$ be $\tau$-structures with universes $A$ and $B$ respectively, and let $\varphi$ be a $\UDL[\tau]$-formula.
    If $\calA\models _T\varphi$ and $\calB$ is a substructure of $\calA$, then $\tmodels{\calB}{T\upharpoonright B}{\varphi}$.
\end{lemma}
\begin{proof}
    The claim is proved using induction on $\varphi$.
    We consider only the case $\varphi = \forall x\theta$, where $\theta$ is quantifier-free.
    Assume $\tmodels{\calA}{T}{\forall x \theta}$.
    Then $\tmodels{\calA}{T^x_{A}}{\theta}$.
    By the induction hypothesis, we have that $\tmodels{\calB}{T^x_{A}\upharpoonright B}{\theta}$.
    Note that
    \[
        T^x_{A}\upharpoonright B=(T\upharpoonright B) ^x_{B},
    \]
    hence $\tmodels{\calB}{T\upharpoonright B}{\forall x\theta}$.
\end{proof}

%\arne{emphasise in the next statement that we are working with infinite structures because otherwise Łoś-Tarski is not applicable.}

Recall the following early result of classical model theory.

% \timon{define the preservation under substructures property?}

\begin{theorem}[Łoś-Tarski preservation theorem, \cite{MR72071,MR66303}]
    A first-order sentence is preserved under substructures if and only if it is logically equivalent to a universal sentence.
\end{theorem}
For the next theorem it is important to recall that coherence is defined not only for finite, but over all (finite or infinite) $\tau$-structures.
\begin{theorem}\label{thm:FO rewrite implies coherent}
    Suppose that $\varphi(x_1,\ldots,x_n)$ is a $\UDL[\tau]$-formula, where $\tau$ is a relational vocabulary.
    If $\varphi$ is first-order rewritable, then $\varphi$ is coherent.
\end{theorem}
\begin{proof}
    Let $\varphi^*$ be a $\FO$-sentence witnessing the first-order rewritability of $\varphi$.
    We claim that $\varphi^*$ is equivalent to a universal $\FO$-sentence.
    By Lemma \ref{lem:closed under substructures}, $\varphi^*$ is closed under substructures.
    Hence, by the Łoś-Tarski preservation theorem, $\varphi^*$ is equivalent to a universal $\FO$-sentence, i.e., to a $\FO$-sentence of the form $\forall x_1 \dots \forall x_m \theta(x_1, \dots, x_m)$, where $m \geq 1$ and $\theta$ is a quantifier-free $\FO$-formula.

	We will show that $\varphi$ is $m^n$-coherent.
    By the downward closure of $\varphi$, it suffices to show the following: if $\calA$ is a $\tau$-structure and $T$ is a team with domain $\{x_1, \ldots, x_n\}$ such that $\calA \not \models_T \varphi$, then there is a subteam $S$ of $T$ such that $|S|\leq m^n$ and $\tnotmodels{\calA}{S}{\varphi}$.
    Towards this goal, let $\calA$ be a $\tau$-structure and let $T$ be a team such that $\calA \not \models_T \varphi$.
    Since $\varphi$ is first-order rewritable to $\forall x_1 \dots \forall x_m \theta(x_1,\dots,x_m)$, we have that $(\calA,\Rel(T))\not \models \forall x_1\dots \forall x_m \theta(x_1,\dots,x_m)$.
    Therefore, there are elements $a_1,\dots,a_m$ from the domain $A$ of $\calA$ such that $(\calA,\Rel(T))\models \neg \theta(a_1,\dots,a_m)$.
    Let $\calB$ be a substructure of $\calA$ with universe $B = \{a_1,\ldots,a_m\}$.
    Since $\theta$ is quantifier-free, we have that $(\calB, \Rel(T)\upharpoonright B)\models \neg \theta(a_1,\dots,a_m)$.
    % , where $(\calB, \Rel(T)\upharpoonright B)$ is the substructure of $(\calA,\Rel(T))$ generated by $B$.
    Thus, $(\calB, \Rel(T)\upharpoonright B)\not \models \forall x_1\dots \forall x_m \theta(x_1,\dots,x_m)$, hence $\tnotmodels{\calB}{S}{\varphi}$, where $S$ is the subteam of $T$ associated with $\Rel(T) \upharpoonright B$.
    Note that $|S|\leq m^n$, since $|\Rel(T)\upharpoonright B|\leq m^n$.
    By Lemma \ref{lem:closed under substructures}, we have that $\calA \not \models_{S} \varphi$.
    This concludes the proof that $\varphi$ is $m^n$-coherent.
\end{proof}

Note that Theorem~\ref{thm:FO rewrite implies coherent} will also follow from results we show later in the paper.
% Here, we showed a self-contained proof.

By combining Theorem~\ref{coherent-FO-reduction} and Theorem~\ref{thm:FO rewrite implies coherent}, we obtain the following connection between coherence and first-order rewritability for quantifier-free dependence logic.
\begin{corollary}
	Let $\varphi$ be a $\DLQF(\tau)$-formula. The following statements are equivalent:
    \begin{enumerate}
        \item The formula $\varphi$ is coherent.
        \item The formula $\varphi$ is rewritable to a universal $\FO$-sentence.
        \item The formula $\varphi$ is first-order rewritable.
    \end{enumerate}
\end{corollary}

Next, we introduce the notion of \emph{coherence in the finite} and show the equivalence between coherence and coherence in the finite.
It is worth noting that this equivalence does not immediately follow from  Theorem \ref{thm:FO rewrite implies coherent}, as its proof uses the Łoś-Tarski preservation theorem which fails in the finite~\cite{DBLP:journals/jsyml/Tait59}.

\begin{definition}[Coherence in the finite]\label{def:finite coherence}
    Suppose $\varphi(x_1, \dots, x_n)$ is a $\DL(\tau)$-formula and $k \in \oN$.
    We say that $\varphi$ is \emph{$k$-coherent in the finite} if for every finite $\tau$-structure $\calA$ and every team $T$ with domain $\{x_1,\ldots,x_n\}$ and range $A$, the following are equivalent:
    \begin{enumerate}
        \item $\tmodels{\calA}{T}{\varphi}$.
        \item For all $k$-element subteams $S\subseteq T$, it holds that $\tmodels{\calA}{S}{\varphi}$.
    \end{enumerate}
We say that $\varphi$ is \emph{coherent in the finite} if there is a $k\geq 1$ such that $\varphi$ is $k$-coherent in the finite.
\end{definition}

%\arne{I suggest to omit the strict in the next statement as it is already within the SNP definition existing.}
\begin{lemma}\label{lem:SNP substructure}
    Let $\varphi$ be a $\SNP$-sentence over some relational vocabulary $\tau$ and let $\calA$ be a $\tau$-structure over universe $A$.
    Then the following statements are equivalent:
    \begin{enumerate}[1.]
        \item $\calA \models \varphi$.
        \item $\calB \models \varphi$, for every finite substructure $\calB$ of $\calA$.
    \end{enumerate}
\end{lemma}
\begin{proof}
    Suppose $\varphi$ is of the form $\exists \vec{S}\, \forall \vec{x} \psi(\vec{x}, \vec{S})$, where $\psi$ is a quantifier-free formula over $\tau$.
    \begin{description}
        \item[``1. $\Longrightarrow$ 2.'']
        This direction follows easily from the preservation of universal first-order sentences under substructures.
        If $\calA$ satisfies $\varphi$, then there exists an expansion of $\calA$ with relations $\vec{S}$ satisfying the universal part.
        If $\calB$ is a substructure of $\calA$ with universe $B$, then the expansion of $\calB$ with the restrictions of the relations $\vec{S}$ on $B$ continues to satisfy the universal first-order sentence $\forall \vec{x} \psi(\vec{x}, \vec{S})$.
        \item[``2. $\Longrightarrow$ 1.'']
        Let $\calA$ be a $\tau$-structure such that every finite substructure $\calB$ of $\calA$ satisfies $\varphi$.
        We expand the vocabulary $\tau$ with a constant symbol $\underline{a}$ for each element $a$ in the universe of $\calA$.
        Let $\Sigma$ be the following set of quantifier-free sentences over this expanded vocabulary:
        \begin{enumerate}
            \item[(i)] The set $\Sigma$ contains the sentences $\underline{a} \neq \underline{b}$, for every two distinct elements $a, b$ in $A$.
            \item[(ii)] The set $\Sigma$ contains all atomic sentences $R(\underline{a_1}, \dots, \underline{a_k})$ and all negated atomic sentences $\neg R(\underline{b_1}, \dots, \underline{b_n})$ that hold in $\calA$.
            \item[(iii)] The set $\Sigma$ contains the quantifier-free sentence $\psi(\underline{a_1}, \dots, \underline{a_n})$ obtained by replacing variables $x_i$ with constant symbols $\underline{a_i}$, for every $n$-tuple $(a_1, \dots, a_n)$ in $A^n$.
        \end{enumerate}
        We claim that $\Sigma$ is finitely satisfiable.
        Let $\Sigma' \subseteq \Sigma$ be a finite subset.
        Let $\underline{a_1}, \dots, \underline{a_m}$ be the finitely many constant symbols occurring in $\Sigma'$.
        Consider the finite substructure $\calB$ of $\calA$ with universe $\{a_1, \dots, a_m\}$.
        By hypothesis, $\calB \models \varphi$, i.e., $\calB \models \exists \vec{S}\, \forall \vec{x} \psi$.
        Thus, $\calB$ expanded with these constant symbols is a model of $\Sigma'$.

        By the Compactness Theorem of first-order logic, since $\Sigma$ is finitely satisfiable, we have that the entire set $\Sigma$ is satisfiable. Let $\calA^*$ be a model of $\Sigma$. By the construction of $\Sigma$, we have that $\calA^*$ contains an isomorphic copy of $\calA$ as a substructure. Being a substructure of $\calA^*$, this isomorphic copy of $\calA$ satisfies $\varphi$, hence
        $\calA$ satisfies $\varphi$ as well. \qedhere
    \end{description}
\end{proof}

We are now able to show that coherence and coherence in the finite are equivalent.
\begin{theorem}\label{lem:finite coherence = coherence}
    Suppose that $\varphi$ is a $\UDL(\tau)$-formula and $k \in \oN$.
    Then following statements are equivalent:
    \begin{enumerate}
        \item The formula $\varphi$ is $k$-coherent.
        \item The formula $\varphi$ is $k$-coherent in the finite.
    \end{enumerate}
    \end{theorem}
\begin{proof}
    Clearly, coherence implies coherence in the finite.
    Thus, it suffices to show the implication ``\itemstyle{2.} $\Longrightarrow$ \itemstyle{1.}''. Towards a contradiction, assume that there is an infinite model $\calA$ and a team $T$ such that $\tnotmodels{\calA}{T}{\varphi}$, but for all $k$-element subteams $S$ of $T$, it holds that $\tmodels{\calA}{S}{\varphi}$.
    By Proposition~\ref{prop:qf-strict}, $\varphi$ is rewritable to some $\SNP$-sentence $\varphi^*$.
    Since $(\calA, \Rel(T))\not \models \varphi^*$, Lemma~\ref{lem:SNP substructure} implies that there exists a finite substructure $\calB$ of $\calA$ such that $(\calB, \Rel(T)\upharpoonright B) \not \models \varphi^*$. %, where $\calB=(B,\Rel(T)\upharpoonright B)$.
    This implies that $\calB\not \models_{T\upharpoonright B}\varphi$. Furthermore, by Lemma~\ref{lem:closed under substructures}, we have that $\calB \models_{S}\varphi$, for all $k$-element subteams $S$ of $T\upharpoonright B$.
    This contradicts the assumption that $\varphi$ is $k$-coherent in the finite, hence  $\varphi$ is indeed $k$-coherent.
\end{proof}

\section{Strong coherence and rewritability}\label{sec:strong coherence}
Regarding universal dependence logic formulas, there is an alternative and, as it turns out, stricter way to define coherence.
This stricter notion of coherence was introduced and studied in Jarmo Kontinen's doctoral thesis~\cite{JarmoThesis}.

\begin{definition}[Strong coherence]\label{def:strong coherence}
    Suppose $\varphi(x_1, \dots, x_n)$ is a $\UDL(\tau)$-formula of the    form $\forall \vec{x}\theta$, where $\theta \in \DLQF(\tau)$, and suppose $k \in \oN$.
    We say that $\varphi$ is \emph{strongly $k$-coherent} if for every (finite or infinite) $\tau$-structure $\calA$ and every team $T$ with domain $\{x_1,\ldots,x_n\}$ and range $A$, the following are equivalent:
    \begin{enumerate}
        \item $\tmodels{\calA}{T}{\varphi}$.
        \item For all $k$-element subteams
        $S \subseteq T^{\vec{x}}_A$, it holds that $\tmodels{\calA}{S}{\theta}$.
    \end{enumerate}
\end{definition}
The related concepts of \emph{strongly coherent} and \emph{strongly coherent in the finite} are defined as the respective notions for coherence.

Clearly, there is no difference between coherence and strong coherence for $\DLQF$-formulas.
As we shall see in the following, these notions are different for $\UDL$-formulas.
\begin{proposition}\label{prop:strong coherence implies coherence}
    Let $\varphi$ be a $\UDL(\tau)$-formula.
    If $\varphi$ is strongly $k$-coherent, then it is $k$-coherent.
\end{proposition}
\begin{proof}
    For simplicity, assume that $\varphi$ has a single universal quantifier, i.e., it is of the form $\forall x \theta(x,x_1,\dots,x_n)$, where $\theta$ is a $\DLQF(\tau)$-formula.
    Assume further that $\calA$ is a $\tau$-structure and $T$ is a team on $\calA$ such that for every $S\subseteq T$ with $|S|=k$, we have that $\tmodels{\calA}{S}{\varphi}$; by downward closure, we have that $\tmodels{\calA}{S}{\varphi}$ holds also for every $S \subseteq T$ such that $|S| \leq k$.
    To show that $\varphi$ is $k$-coherent, we must show $\tmodels{\calA}{T}{\varphi}$.
    Since $\varphi$ is strongly $k$-coherent, it suffices to show that for every subteam $R$ of $T^x_A$ of size $k$, we have that $\tmodels{\calA}{R}{\theta}$.
    So, let $R\subseteq T^x_A$ be a subteam of $T^x_A$ with $|R|=k$.
    Let $R^*$ be the projection of $R$ on the variables $x_1,\ldots,x_n$.
    Then, $R^*$ is a subteam of $T$ and $|R^*| \leq k$; furthermore, we have that $R\subseteq (R^*)^x_A$.
    By the hypothesis about $\varphi$, we have that $\calA \models_{R^*} \varphi$, which means that $\calA \models_{(R^*)^x_A}\theta$, hence, by downward closure, we have that $\calA \models_R \theta$, which was to be shown.
\end{proof}

The following known result connects FO-rewritability to strongly coherent $\UDL(\tau)$-formulas.
\begin{theorem}[{\cite{JaKontinen13}}]\label{thm:coherence-quantified}
    Suppose $\varphi$ is a strongly coherent $\UDL(\tau)$-formula.
    Then $\varphi$ is first-order rewritable.
    In fact, it is rewritable to a universal first-order sentence $\varphi^*$.
\end{theorem}
The following simple example shows that Theorem \ref{thm:coherence-quantified} does not extend to $k$-coherent $\UDL(\tau)$-formulas;  furthermore, $k$-coherence for some $k$ does not imply strong coherence.

\begin{example}
    Consider the $\UDL(\tau)$-sentence
    \[
        \varphi = \forall x \forall y \forall u \forall v \big(\neg R(x,y,u,v)\vee (\dep(x,y)\vee \dep(u,v))\big).
    \]
    Since $\varphi$ does not have any free variables, it is 1-coherent.
    % by locality.
    In contrast, it is easy to see that $\tmodels{\calA}{\emptyset}{\varphi}$ if and only if $\calA \models_{T} \dep(x,y)\vee \dep(u,v)$, where $T=\{\,s \mid (s(x),s(y),s(u),s(v)) \in R\,\}$.
    Therefore, $\varphi$ is not FO-rewritable and hence it is not strongly $k$-coherent for any $k$.
\end{example}
\begin{theorem}\label{strong_coh}
    Let $\varphi(z_1,\ldots,z_n)
    $ be a $\UDL(\tau)$-formula of the form $\forall y_1 \ldots \forall y_l \beta$.
    The following statements are equivalent:
    \begin{enumerate}
        \item The formula $\varphi$ is strongly coherent.
        \item The formula $\varphi$ is rewritable to a universal $\FO$-sentence.
        \item The formula $\varphi$ is first-order rewritable.
    \end{enumerate}
\end{theorem}
\begin{proof}
    The implication ``\itemstyle{1.} $\Longrightarrow$ \itemstyle{2.}'' follows from Theorem \ref{thm:coherence-quantified}.
    The implication ``\itemstyle{2.} $\Longrightarrow$ \itemstyle{3.}'' is trivial.
    We show the implication ``\itemstyle{3.} $\Longrightarrow$ \itemstyle{1.}''.

    Let $\varphi^*$ be a $\FO$-sentence witnessing the first-order rewritability of $\varphi$.
    By Lemma \ref{lem:closed under substructures}, $\varphi^*$ is closed under substructures.
    Hence, by the Łoś-Tarski preservation theorem, $\varphi^*$ is equivalent to a universal $\FO$-sentence, i.e., to a $\FO$-sentence of the form
    $\forall x_1\ldots \forall x_m \theta(R)$, where $m\geq 1$ and $\theta$ is a quantifier-free $\FO$-formula.

	We will show that $\varphi$ is $m^{n+l}$-coherent.
    By the downward closure of $\varphi$, it suffices to show the following:
    if $\calA$ is a $\tau$-structure and $T$ is a team with domain $\{z_1,\ldots,z_n\}$ such that $\calA \not \models_T \varphi$, then there is a subteam $S$ of $T^{\vec{y}}_A$ such that $|S|\leq m^{n+l}$ and $\tnotmodels{\calA}{S}{\beta}$.
    Towards this goal, let $\calA$ be a $\tau$-structure and let $T$ be a team such that $\tmodels{\calA}{T}{\varphi}$.
    Since $\varphi$ is first-order rewritable to $\forall x_1\ldots \forall x_m \theta(R)$, we have that $(\calA,\Rel(T))\not \models \forall x_1\ldots \forall x_m \theta(R)$.
    Therefore, there are elements $a_1,\ldots,a_m$ from the universe $A$ of $\calA$ such that $(\calA,\Rel(T))\models \neg \theta(a_1,\ldots,a_m)$.
    Let $\calB$ be a substructure of $\calA$ with universe $B = \{a_1,\ldots,a_m\}$.
    Since $\theta$ is quantifier-free, we have that $(\calB, \Rel(T)\upharpoonright B)\models \neg \theta(a_1,\ldots,a_m)$.
    %, where $(\calB, \Rel(T)\upharpoonright B)$ is the substructure of $(\calA,\Rel(T))$ generated by $B$.
    Thus, $(\calB, \Rel(T)\upharpoonright B) \not\models \forall x_1\ldots \forall x_m \theta(x_1,\ldots,x_m)$, hence $\tnotmodels{\calB}{S}{\varphi}$, where $S$ is the subteam of $T$ associated with $\Rel(T)\upharpoonright B$.
    By the semantics of the universal quantifier we get that $\tnotmodels{\calB}{S^{\vec{y}}_B}{\theta}$.
    Now, note that $|S^{\vec{y}}_B|\leq m^{n+l}$, since $|\Rel(T)\upharpoonright B| \leq m^n$.
    By Lemma \ref{lem:closed under substructures}, we have that $\tnotmodels{\calA}{S^{\vec{y}}_B}{\beta}$.
    This concludes the proof that $\varphi$ is strongly $m^{n+l}$-coherent.
\end{proof}
Note that, together with Proposition \ref{prop:strong coherence implies coherence}, Theorem \ref{strong_coh} implies  Theorem \ref{thm:FO rewrite implies coherent}.

As the translation to SNP-sentences extends to  universal $\DL$-formulas, the following holds.

%\arne{is it important that we miss the $(\tau)$ in the next expression for $\UDL$?}
\begin{theorem}
    Suppose $\varphi$ is a $\UDL(\tau)$-formula and $k \in \oN$.
    Then following statements are equivalent:
    \begin{enumerate}
        \item The formula $\varphi$ is strongly $k$-coherent.
        \item The formula $\varphi$ is strongly $k$-coherent in the finite.
    \end{enumerate}
    \end{theorem}
% Nicolas: fix notation in the following proof
\begin{proof}
    Clearly, strong coherence implies strong coherence in the finite.
    Thus, it suffices to show the implication ``\itemstyle{2.} $\Longrightarrow$ \itemstyle{1.}''.
    Towards a contradiction, assume that $\varphi$ is strongly $k$-coherent in the finite and that there is an infinite $\tau$-structure $\calA$ and a team $T$ of $\calA$ such that $\tnotmodels{\calA}{T}{\varphi}$, but for all $k$-element subteams $S$ of $T^{\vec{x}}_A$, it holds that $\tmodels{\calA}{S}{\theta}$.
    By Proposition~\ref{prop:qf-strict}, $\varphi$ is rewritable to some $\SNP$-sentence $\varphi^*$.
    Since $(\calA, \Rel(T)) \not\models \varphi^*$, Lemma \ref{lem:SNP substructure} implies that there exists a finite substructure $\calB$, with $B \subseteq A$, such that $(\calB,\Rel(T)\upharpoonright B) \not\models \varphi^*$.
    This implies that $\tnotmodels{\calB}{T \upharpoonright B}{\varphi}$.
    Let $S$ be a $k$-element subteam of $(T \upharpoonright B)^{\vec{x}}_B$.
    Note that $S$ is a $k$-element subteam of $T^x_A$, hence, by our assumption about $\calA$ and $\varphi$, we have that $\tmodels{\calA}{S}{\theta}$.
    Therefore, by Lemma \ref{lem:closed under substructures}, we have that $\tmodels{\calB}{S}{\theta}$.
    Since $\varphi$ is strongly $k$-coherent in the finite, we have that $\tmodels{\calB}{T \upharpoonright B}{\varphi}$ and, thus, we have arrived at a contradiction.
%Now the existence of finite $\calA\upharpoonright B$ and $ X\upharpoonright B$ contradicts the assumption that $\varphi$ is strongly $k$-coherent in the finite.
\end{proof}

We close this chapter with two further results on coherence of quantified formulas. We first generalize coherence from finite to infinite cardinals.

\begin{definition}[Cardinal coherence]\label{def:cardinal-coh}
    Suppose $\varphi(x_1, \dots, x_n)$ is a $\DL(\tau)$-formula and $\lambda$ an arbitrary cardinal number.
    Then $\varphi$ is \emph{$\lambda$-coherent} if  for every $\tau$-structures $\calA$ and every team $T$ with domain $\{x_1,\ldots,x_n\}$ and range $A$, the following are equivalent:
    \begin{enumerate}
        \item $\tmodels{\calA}{T}{\varphi}$.
        \item For $Y\subseteq T$, such that $|Y|<\lambda$ it holds that $\calA \models_Y \varphi$.
    \end{enumerate}
\end{definition}

The next example shows that there are formulas with the existential quantifier which are cardinal incoherent, i.e, there is no cardinal $\lambda$ such that the formulas are $\lambda$-coherent. This example is given also by Ciardelli and Grilletti~\cite{CIARDELLI2022103155} in the context of inquisitive first-order logic.
\begin{example}\label{ex:cardinal incoherence}
    The following dependence formula is not $\lambda$-coherent for any $\lambda$:
    \[
        \varphi = \exists z(\dep(z)\wedge \neg z=x).
    \]
    Indeed, if $A$ is any non-empty set and $T$ is the team of all assignments $s$ with domain $\{x\}$ and range $A$, then $T$ does not satisfy $\varphi$ in the model $\calA$ (of the empty vocabulary) but any proper subteam does.
    Namely, if $S \subsetneq T$ misses $t\in T$, then $\tmodels{\calA}{S}{\varphi}$, as $t(x)\in A$ gives a value that can serve as the interpretation of $z$ in the formula $\varphi$.
\end{example}

The next theorem shows that, contrary to the previous example, universal formulas are always $\omega$-coherent.

\begin{theorem}
    Suppose $\varphi$ is a $\UDL$-formula.
    Then $\varphi$ is $\omega$-coherent.
\end{theorem}
\begin{proof}
    Assume $\calA$ and $T$ are such that $\tnotmodels{\calA}{T}{\varphi}$.
    By Lemma \ref{lem:SNP substructure}, there exists a finite substructure $\calB$ of $\calA$ such that $\tnotmodels{\calB}{T\upharpoonright B}{\varphi}$.
    By Lemma \ref{lem:closed under substructures}, we must have $\tnotmodels{\calA}{T\upharpoonright B}{\varphi}$ which shows that claim.
\end{proof}

\section{On the complexity of deciding coherence}\label{sec:complexity}
In this section, we explore algorithmic aspects of coherence in dependence logic.
We begin by introducing the main decision problem concerning coherence.

\begin{definition} \label{def:coh-l-problem}
    Let $\mathcal L$ be a collection of $\DL$-formulas.
    \begin{itemize}
    \item The \emph{\COHERENCE problem for $\mathcal L$} is the following decision problem: given an $\mathcal L$-formula $\varphi$ and a positive integer $k$, is $\varphi$ $k$-coherent?
    % \timon{$k$-coherent? Otherwise there is no need for a $k$ in the input.}
    \item Let $k$ be a positive integer. The \emph{\kCOHERENCE problem for $\mathcal L$} is the following decision problem: given an $\mathcal L$-formula $\varphi$, is $\varphi$ $k$-coherent?
    \end{itemize}
    The \emph{\COHERENCE problem for $\mathcal L$ in the finite} and the \emph{\kCOHERENCE problem for $\mathcal L$ in the finite} are defined in an analogous manner.
\end{definition}

Clearly, a complexity-theoretic upper bound for \COHERENCE implies the same upper bound for \kCOHERENCE, for every fixed $k\geq 1$. Similarly, a complexity-theoretic lower bound for \kCOHERENCE for some $k\geq 1$ implies the same lower bound for \COHERENCE.
The next result shows that deciding \kCOHERENCE for $\DL$-formulas is a highly undecidable problem.

\begin{theorem} \label{thm:coh-non-arithm} \label{thm:k-coherent-non-arith}
    For every fixed natural number $k\geq 2$, the \kCOHERENCE problem for dependence logic $\DL$ is of non-arithmetical complexity.
\end{theorem}
\begin{proof}
    Consider the vocabulary $\tau$ consisting of just one binary predicate symbol.
    Let $\theta$ be the first-order sentence $\forall x_1\ldots\forall x_{k+1}\bigvee_{1\le i<j\le k+1}x_i=x_j$.
    We show that for all dependence logic sentences $\varphi$ over the vocabulary $\tau$, we have that $\varphi$ is valid in models of size $> k$ if and only if the formula $\psi\coloneqq(\theta\vee\varphi)\vee \bigvee_{i=1}^k\dep(x)$ is $k$-coherent.
    \begin{description}
    \item[``$\Longrightarrow$''] Suppose $\varphi$ is valid in models of size $> k$.
    Then $\theta\vee\varphi$ is valid, for given any model of the vocabulary $\tau$, either it has at most $k$ elements, and thereby satisfies $\theta$, or it has at least $k+1$ elements, and thereby satisfies $\varphi$.
    Hence $(\theta\vee\varphi)\vee \bigvee_{i=1}^k\dep(x)$ is $k$-coherent for all $k$, as the formula is valid.
    \item[``$\Longleftarrow$''] Suppose $M\not\models\varphi$, where $M$ has at least $k+1$ elements.
    Then $(\theta\vee\varphi)\vee \bigvee_{i=1}^k\dep(x)$ is not $k$-coherent because in the model $M$ it is equivalent to (i.e., satisfied by the same teams as) $\bigvee_{i=1}^k\dep(x)$, which is not $k$-coherent.
    Since validity in models of size $>k$ is not arithmetical (Väänänen~\cite[p. 150]{Vaananen-book}), we are done. \qedhere
\end{description}
\end{proof}
In fact, validity in dependence logic is $\Pi_2$-complete, hence much worse than non-arithmetical~\cite[p.\ 150]{Vaananen-book}.

The next theorem concerns  coherence in the finite.  Recall that a decision problem $Q$ is
\emph{recursively enumerable (r.e.)} if there is an algorithm that enumerates the ``yes'' inputs to this problem. We write $\RE$ for the collection of all recursively enumerable decision problems. We also write $\coRE$ for the collection of all decision problems whose complement is in $\RE$.
Furthermore, a problem $Q$ is \emph{$\RE$-complete} if it is in $\RE$, and every problem in $\RE$ can be reduced to $Q$ via a many-one recursive reduction. The notion of $\coRE$-completeness is defined analogously.

%further, $Q$ is \emph{r.e.-complete} if every r.e.\ decision problem can be reduced to $Q$ via a \textcolor{red}{many-one recursive reduction}. %Turing
%eduction
 %We say that $Q$ is \emph{co-r.e.}
%\emph{$\coRE$-complete}
%if the complement of $Q$ is recursively enumerable.
% We write $\coRE$ for the collection of all co-r.e. decision problems.

%or \emph{co-r.e.-complete} or r.e.-complete
\begin{theorem} \label{thm:coh-und-finite}
    The following statements are true about dependence logic $\DL$.
    \begin{enumerate}
    \item The \COHERENCE problem for $\DL$  in the finite is {in}  $\coRE$.
    \item For every fixed $k\geq 2$, the \kCOHERENCE problem for $\DL$ in the finite is $\coRE$-complete.
    In fact, it is $\coRE$-complete  even for $\DL$-formulas of the form $\psi_1\vee\psi_2$, where $\psi_1$ is a $\FO$-formula and $\psi_2$ is a disjunction of $k$ constancy atoms.
    \end{enumerate}
\end{theorem}
%\nicolas{where do $\psi_1$ and $\psi_2$ appear in the proof?}\jouko{Corrected it.}
\begin{proof}
    Let $\theta$ be the first-order sentence $\forall x_1\ldots\forall x_{k+1} \bigvee_{1\le i<j \le k+1} x_i = x_j$.
    As in the proof of Theorem~\ref{thm:coh-non-arithm}, we may observe that if $\varphi$ is a first-order sentence, then $\varphi$ is valid in finite models of size $> k$ if and only if the formula $\psi\coloneqq\theta\vee\varphi\vee \bigvee_{i=1}^k\dep(x)$ is $k$-coherent.
    Since first-order validity in finite models of size $> k$ is a $\coRE$-complete problem (this result is known as Trakhtenbrot's theorem~\cite{MR33784}), we are done.
\end{proof}
With regards to universal $\DL$-formulas, Theorem  \ref{lem:finite coherence = coherence} tells that $k$-coherence coincides with $k$-coherence in the finite. Thus, from Theorems \ref{lem:finite coherence = coherence}
and \ref{thm:coh-und-finite}, we obtain the following result.

\begin{corollary} \label{thm:coh-core}
    The \COHERENCE problem for $\UDL$ is in $\coRE$.
\end{corollary}
Since \COHERENCE for $\UDL$ is in $\coRE$, it follows that for every fixed $k\geq 2$, the \kCOHERENCE problem for $\UDL$ is also in $\coRE$.
It is an open question whether anyone of these problems is decidable or undecidable.
In fact, the question is open even for the collection $\DLQF$ of quantifier-free dependence logic formulas.

\subsection{Coherence in propositional dependence logic}\label{sec:prop-dep-logic}
% \nicolas{Do we still want to state the model checking, validity and implication problem results somewhere in the beginning of this section?}
In this section, we consider the coherence problem in the context of propositional dependence logic $\PDL$.
%or, equivalently,
%for quantifier-free formulas of dependence logic $\DL$ over the empty vocabulary and
%for teams over a binary universe\juha{In PDL we have atoms $p$ and $\neg p$ which could be expressed in FO by $x=1/0$, so I think we need the constants 0,1 or a unary relation symbol P to make the connection between PDL and QF-DL.} Our main result pinpoints the computational compexity of this problem.
We begin by introducing the syntax and semantics of $\PDL$.
The formulas of $\PDL$ are defined by the grammar
\[
	\psi \coloneqq
    p \mid
    \lnot p \mid
    \dep(\vec p,q) \mid
    \psi \land \psi \mid
    \psi\lor\psi,
\]
where $\vec p = p_1, \dots, p_n$ is a vector of propositional variables and $p, q$ are single propositional variables.

A \emph{propositional team} is a team $T$ with domain a set $V$ of propositional variables and range $\{0,1\}$.
We now revisit the semantics of dependence logic from a propositional point of view.
Notice that since the model in $\PDL$ is fixed, we drop it from the satisfaction relation and write $T \models \varphi$ instead of $\tmodels{\{0, 1\}}{T}{\varphi}$.
\begin{definition}
    For a $\PDL$-formula $\varphi$ and a propositional team $T$ whose domain contains the set of propositional variables occurring in $\varphi$, the satisfaction relation $T \models \varphi$ is defined as follows, where $\vec p=p_1,\dots,p_n$ is a vector of propositional variables, and $p,q$ are single propositional variables:
    \begin{tabbing}
    $T \models \dep(\vec y,z)$\;\=\;if\;\=\kill
    $T \models p$ \> if \> for all $s\in T$, we have $s(p)=1$;\\
    $T \models \lnot p$ \> if \> for all $s\in T$, we have $s(p)=0$;\\
    $T \models \dep(\vec p,q)$ \> if \> for all $s_1,s_2\in T$: $s_1(p_j) = s_2(p_j)$ for all $j$ implies $s_1(q)=s_2(q)$;\\
    $T \models \varphi_1 \land \varphi_2$ \> if \> $T \models \varphi_1$ and $T \models \varphi_2$;\\
    $T \models \varphi_1 \lor \varphi_2$ \> if \> there are teams $T_1$, $T_2$ with $T= T_1\cup T_2$ and $T_1 \models \varphi_1$ and $T_2 \models \varphi_2$.
    \end{tabbing}
\end{definition}

We now introduce the notion of coherence for $\PDL$ formulas.
\begin{definition}[Propositional coherence] \label{def:prop-coh}
    Suppose $\varphi(p_1, \dots, p_n)$ is a $\PDL$-formula and $k \in \oN$.
    We say that $\varphi$ is \emph{$k$-coherent} if for every propositional team $T$ with domain $\{p_1, \ldots, p_n\}$, the following are equivalent:
    \begin{enumerate}
        \item $T \models\varphi$.
        \item For every $k$-element subteam $S \subseteq T$, it holds that $S \models \varphi$.
    \end{enumerate}
    We say that $\varphi$ is \emph{coherent} if there is a $k\geq 1$ such that $\varphi$ is $k$-coherent.
\end{definition}

Let $\varphi(p_1,\ldots,p_n)$ be a $\PDL$-formula with $n$ variables.
Since every propositional team with domain $\{p_1,\ldots,p_n\}$ has size at most $2^n$, we have that $\varphi(p_1,\ldots,p_n)$ is $2^n$-coherent, hence it is $k$-coherent for every $k\geq 2^n$.
Thus, the \kCOHERENCE problem for $\PDL$-formulas with $n$ variables is interesting only when $k \leq 2^n-1$.

% \timon{currently, the \COHERENCE problem asks for $k$-coherence with $k$ in the input. But here, we mean the problem of asking if $\varphi$ is coherent at all, right?}

\begin{example}
Note that, for each $k\geq 2$, there is a $\PDL$-formula $\theta_k$ with some number of variables such that $\theta_k$ is not $k$-coherent, hence it is not $k'$-coherent for any $k'\leq k$.
    For $k=3$, we take $\theta_k$ to be the $\PDL$-formula $\dep(p,q) \lor \dep(p,r)$.
    Clearly, every team $T$ with $3$ elements satisfies $\theta_3$, but the team $\Rel(T) = \{(0,0,0), (0,0,1), (0,1,0), (0,1,1)\}$, which has $4$ elements, does not satisfy $\theta_3$.
\end{example}
More generally, it can be shown (cf.\ Fröhlich~et~al.~\cite[Thm. 25 (1.)]{FKM2026}) that the formula
\[
    \dep(p_1,\ldots,p_m,  q_1\ldots,q_m) \lor \dep(q_1,\ldots,q_m,  p_1,\ldots,p_m)
\]
is not $2^{m+1}$-coherent, where $\dep(p_1,\ldots,p_m,  q_1\ldots,q_m)$ abbreviates $\bigwedge_{i=1}^m \dep(p_1,\ldots,p_m, q_i)$.

The next result establishes membership of \kCOHERENCE for $\PDL$ in $\PiE{2}$, that is, in the second level of the exponential-time hierarchy (Mocas \cite{DBLP:journals/tcs/Mocas96}).
Formally, $\PiE{2}$ is defined as $\coNEXPTIME^{\NP}$, i.e., the complement of the class of languages decidable by a non-deterministic Turing machine with exponential time and $\NP$-oracle.
\begin{theorem} \label{thm:coh-mem}
    For every fixed $k \geq 2$, the \kCOHERENCE problem for $\PDL$ is in $\PiE{2}$.
\end{theorem}
\begin{proof}
    Suppose $\varphi \in \PDL$ contains $n$ variables.
    An algorithm deciding \kCOHERENCE for $\varphi$ first universally guesses a propositional team $T$ over $n$ variables.
    % \miika{Universally guesses?}
    This team is of size $|T| \leq 2^n$.
    If $|T| > k$ and $T \not\models \varphi$, then the algorithm searches through all $S \subseteq T, |S| = k$ and accept if and only if there is a $S \not\models \varphi$.
    There are at most $2^{nk}$ many subteams of size $k$.
    These steps can be done in $\coNEXPTIME$, except for the model checking of $T$ and for the subteams.
    Ebbing and Lohmann showed model checking for $\PDL$ to be $\NP$-complete~\cite{DBLP:conf/sofsem/EbbingL12}, thus the algorithm requires an $\NP$-oracle and the problem is in $\PiE{2}$.
\end{proof}

Before proceeding further, we recall the notions of polynomial-time many-one reductions and polynomial-time truth-table reductions. Every decision problem $P$ can be identified with its \emph{characteristic function $\chi_P$}, where for every instance $I$ of  $P$, we have that
\[
    \chi_P(I)=
    \begin{cases}
    1, & \text{if $I$ is a ``yes'' instance of $P$, and}\\
    0, & \text{if $I$ is a ``no'' instance of $P$.}
    \end{cases}
\]
\begin{definition}
    Let $P$ and $Q$ be two decision problems.
    \begin{itemize}
    \item We say that \emph{$R$ is polynomial-time many-one reducible to $Q$}, denoted
    $P\preceq^p_m Q$, if there is polynomial-time algorithm $A$ such that for every input $I$ of $P$, the algorithm $A$ produces an input $J$ of $Q$ such that $\chi_P(I) =1 $ if and only if $\chi_Q(J)= 1$.
    \item We say that \emph{$R$ is polynomial-time truth-table reducible to $Q$}, denoted $P \preceq^p_{tt} Q$, if there is polynomial-time algorithm $A$ such that for every input $I$ of $P$, the algorithm $A$ produces a Boolean formula $\theta(z_1,\ldots,z_m)$ and inputs $J_1,\ldots,J_m$ of $Q$ such that $\chi_P(I) =1 $ if and only if $\theta(\chi_Q(J_1),\ldots,\chi_Q(J_m))=1$.
    \item We say that \emph{$R$ is polynomial-time bounded-truth-table reducible to $Q$}, denoted $P \leqbtt Q$, if there is polynomial-time algorithm $A$ and a Boolean formula $\theta(z_1,\ldots,z_m)$ such that for every input $I$ of $P$, the algorithm $A$ produces inputs $J_1,\ldots,J_m$ of $Q$ such that $\chi_P(I) = 1$ if and only if $\theta(\chi_Q(J_1),\ldots,\chi_Q(J_m)) = 1$.
    \end{itemize}
\end{definition}
Clearly, $P \preceq^p_m Q$ implies $P \leqbtt Q$, while $P\leqbtt Q$ implies $P\preceq^p_{tt} Q$.

We now show that \kCOHERENCE for $\PDL$ is $\NEXPTIME$-hard, for each fixed $k\geq 2$.

\begin{theorem}\label{thm:coh-nexptime-hard}
    For every fixed $k\geq 2$, the \kCOHERENCE problem for $\PDL$ is $\NEXPTIME$-hard via polynomial-time many-one reductions.
\end{theorem}
\begin{proof}
    Fix a positive integer $k\geq 2$.
    Virtema~\cite{DBLP:journals/iandc/Virtema17} showed that the validity problem for $\PDL$ is $\NEXPTIME$-hard via polynomial-time reductions.
    We will show that the validity problem for $\PDL$ is polynomial-time reducible to the \kCOHERENCE problem for $\PDL$.

    Let $\theta_k$ be a $\PDL$-formula that is not $k$-coherent and let $q_1,\ldots,q_m$ be the variables of $\theta_k$.
    Let $\varphi(p_1,\ldots,p_n)$ be a $\PDL$-formula which we want to determine whether or not it valid.
    Without loss of generality, we may assume that the variables $p_1,\ldots,p_n$ of $\varphi$ are different from the variables $q_1,\ldots,q_m$ of $\theta_k$.
    Construct the $\PDL$-formula $\varphi \vee \theta_k$.
    We claim that $\varphi$ is valid if and only if $\varphi \lor \theta_k $ is $k$-coherent.

    Assume that $\varphi$ is valid.
    If $T$ is a team with domain $\{p_1,\ldots,p_n, q_1,\ldots,q_m\}$, then $T$  satisfies $\varphi$, hence $T$ satisfies $\varphi \vee \theta_k$.
    Therefore, $\varphi \vee \theta_k$ is valid, hence it is $k$-coherent (since it is $1$-coherent).

    Next, assume that $\varphi$ is not valid.
    Hence, there is a non-empty team $S$ with domain $\{p_1,\ldots,p_n\}$ that does not satisfy $\varphi$.
    Since $\theta_k$ is not $k$-coherent, there is a non-empty team
    $T$ with domain $\{q_1,\ldots,q_m\}$ that witnesses this fact, i.e., $T\not \models \theta_k$ but if $T'\subseteq T$ is a subteam of $T$ of size at most $k$, then $T'\models \theta_k$.
    We will show that the team $S\times T$ witnesses that $\varphi \lor \theta_k$ is not $k$-coherent.
    First, we show that $S \times T \not\models \varphi \vee \theta_k$.
    For this, we overload the notation for restrictions as follows: if $U \subseteq V$, then the \emph{restriction of $T$ on $U$} is the team $T\upharpoonright U = \{\,s\upharpoonright U \mid s\in T\,\}$, where $s\upharpoonright U$ is the restriction of the assignment $s$ to $U$.
    Assume that $Y$ and $Z$ are subteams of $S\times T$ such that $Y\cup Z=S\times T$.
    We now distinguish two cases: if $Z \upharpoonright \{q_1,\ldots,q_m\} = T$, we have that $Z \not\models \theta_k$; if $Z\upharpoonright \{q_1,\ldots,q_m\} \subsetneq T$, then $Y\upharpoonright \{p_1,\ldots,p_n\} = S$ and so we have that $Y\not \models \varphi$.
    Therefore, $S\times T \not \models \varphi \vee \theta_k$.
    Next, assume that $W$ is a subteam of $S\times T$ of size at most $k$.
    Split the team $W$ to $Y=\emptyset$ and $Z = W$. Then $Y\models \varphi$ and $Z\models \theta_k$, since $Z\upharpoonright \{q_1,\ldots,q_n\}$ is a subteam of $T$ of size at most $k$.
\end{proof}

Our next goal is to improve the lower bound in the preceding Theorem \ref{thm:coh-nexptime-hard}. For this, we will exhibit a reduction from the entailment problem in propositional dependence logic, which was shown to be $\PiE{2}$-complete by Hannula~\cite{Hannula19}.

% \commentout{
% \begin{proposition} \label{prop:coher-conj}
% Let $k\geq 2$ and let $\varphi$
% of propositional dependence logic $\PDL$.  If $\varphi$ is $k$-coherent and $\varphi \models \psi$, then the conjunction $\varphi \wedge \psi$ is $k$-coherent.
% \end{proposition}
% \begin{proof}
% %We give the proof for the case in which $\varphi$
% %and $\psi$ are $\PDL$-formulas. The case for $\DL$-formulas is entirely analogous.
% Let $T$ be a team such that for every subteam $T'\subseteq T$ with $|T'|\leq k$, we have that $T'\models \varphi \wedge \psi$. We must show that $T\models \varphi \wedge \psi$. By the semantics of conjunction of $\DL$-formulas, we have that
% $T'\models \varphi$, for every subteam $T'$ of $T$ with
% $|T'|\leq k$. Since $\varphi$ is $k$-coherent, we have that
% $T\models \varphi$. Since $\varphi \models \psi$, we have that
% $T\models \psi$, hence $T\models \varphi \wedge \psi$, which was to be shown.
% \end{proof}
% }

\begin{definition}
    Let $k$ be a positive integer and let $\varphi$ and $\psi$ be two $\PDL$-formulas.
    We say that $\varphi$ \emph{$k$-entails $\psi$}, denoted $\varphi\models_k \psi$, if for every team $T$ of size at most $k$ that satisfies $\varphi$, it holds that $T$ also satisfies $\psi$.
\end{definition}

\begin{proposition} \label{prop:log-impl}
    Let $k$ be a positive integer, $\varphi$ be a $k$-coherent $\PDL$-formula, and $\psi$ be a $\PDL$-formula.
    Then the following statements are equivalent:
    \begin{enumerate}
    \item $\varphi \models \psi$.
    \item $\varphi \wedge \psi$ is $k$-coherent and $\varphi \models_k \psi$.
    \end{enumerate}
\end{proposition}
\begin{proof}
    The implication ``\itemstyle{1.} $\Longrightarrow$ \itemstyle{2.}'' follows from the following two facts: (i) if $\varphi\models \psi$, then $\varphi\wedge \psi$ is logically equivalent to $\varphi$, hence $\varphi\land\psi$ is $k$-coherent; (ii) if $\varphi \models \psi$, then $\varphi\models_k \psi$, for every $k\geq 1$.

    Towards the implication ``\itemstyle{2.} $\Longrightarrow$ \itemstyle{1.}'', assume that $\varphi \wedge \psi$ is $k$-coherent and that $\varphi\models_k \psi$.
    Let $T$ be a team such that $T\models \varphi$.
    We must show that $T\models \psi$.
    By downward closure, we have that $S \models \varphi$ holds, for every subteam $S \subseteq T$ with $|S|\leq k$.
    Since $\varphi \models_k \psi$, it follows that $S\models \psi$ holds, for every subteam $S\subseteq T$ with $|S|\leq k$.
    Therefore, we have that $S\models \varphi \wedge \psi$ holds, for every subteam $S\subseteq T$ with $|S|\leq k$.
    Since $\varphi \wedge \psi$ is $k$-coherent, it follows that $T\models \varphi \wedge \psi$, hence $T\models \psi$, which was to be shown.
\end{proof}

\begin{theorem}\label{thm:coh-hard}
    For every fixed $k\geq 2$, the \kCOHERENCE problem for $\PDL$ is $\PiE{2}$-complete via  polynomial-time bounded-truth-table reductions.
\end{theorem}
\begin{proof}
    Fix a $k\geq 2$.
    Membership of \kCOHERENCE for $\PDL$ in $\PiE{2}$ follows from the preceding Theorem \ref{thm:coh-mem} as $k$ is fixed.
    For $\PiE{2}$-hardness, we will use Theorem \ref{thm:coh-nexptime-hard} and a result from Hannula~\cite{Hannula19} concerning the entailment problem for $\PDL$.
    Specifically, as shown by Hannula~\cite{Hannula19}, the following problem is $\PiE{2}$-hard via polynomial-time reductions:
    given a finite conjunction $\varphi$ of dependence atoms and a $\PDL$-formula $\psi$, does $\varphi \models \psi$?
    Since every dependence atoms is $2$-coherent and since $2$-coherence is preserved under conjunctions, we have that $\varphi$ is a $2$-coherent $\PDL$-formula, hence it is also $k$-coherent for every $k\geq 3$.
    Thus, by Proposition \ref{prop:log-impl}, if $\varphi$ is a finite conjunction of dependence atoms and $\psi$ is a $\PDL$-formula, then
    \[
        \varphi\models \psi\text{ if and only if }\varphi\land \psi\text{ is $k$-coherent and }\varphi\models_k \psi.
    \]
    We now show that, for every fixed $k\geq 2$, the $k$-entailment problem is in $\Pi_2^{\mathsf P}$.
    Indeed, by definition, we have that $\varphi \models_k \psi$ if and only if every team $T$ of size at most $k$ that satisfies $\varphi$ also satisfies $\psi$.
    Hence, to verify that $\varphi \not \models_k \psi$, we guess a team $T$ of size at most $k$ (which is a polynomial-sized guess in the size of $\varphi$ and $\psi$) and then check that $T\models\varphi$ but $T\not \models \psi$.
    Since the model-checking problem for $\PDL$ in $\NP$, it follows that $k$-entailment is in $\Pi_2^{\mathsf P}$.

    The preceding analysis shows that the entailment problem for $\PDL$ has a polynomial-time reduction to the conjunction of the \kCOHERENCE problem for $\PDL$ with a problem in $\Pi_2^{\mathsf P}$.
    Since $\Pi_2^{\mathsf P}\subseteq \NEXPTIME$ and since, by Theorem \ref{thm:coh-nexptime-hard}, the \kCOHERENCE problem for $\PDL$ is $\NEXPTIME$-hard via polynomial-time reductions, we have that the entailment problem for $\PDL$ has a polynomial-time truth table reduction to the \kCOHERENCE problem.
    %We  obtain the following
    %$\Pi_2^{\mathsf P}$-reduction of the entailment problem $\varphi\models \psi$ for $\PDL$, where $\varphi$ is a finite conjunction of dependence atoms, to the \kCOHERENCE problem for $\PDL$:
    %\begin{enumerate}
    %\item First, use a $\Pi^{\mathsf P}_2$-check to determine whether $\varphi\models_k \psi$. If the answer is ``no'', then report that $\varphi \not \models \psi$. Otherwise,
    %\item Check if $\varphi\land \psi$ is $k$-coherent. If the answer is ``yes'', then report that $\varphi  \models \psi$; otherwise, report that $\varphi \not \models \psi$.
    %\end{enumerate}
\end{proof}
It remains open, whether the preceding result can be strengthened to hold for polynomial-time many-one reductions as well.
One way to prove such a strengthening would be to show that for every fixed $k\geq 2$, there exists a polynomial-time algorithm which, given two $\PDL$-formulas $\phi$ and $\psi$ produces a $\PDL$-formula $\theta$, such that both $\phi$ and $\psi$ are $k$-coherent if and only if $\theta$ is $k$-coherent.

Our lower bound for the complexity of deciding $k$-coherence for $\PDL$-formulas also implies the following result on the minimal size of witnesses of the failure of $k$-coherence.

%\juha{Please check the below argument and the bounds. I cannot find the picure in which we formulated this with Phokion earlier.}

%The lower bound for the complexity of $k$
%As a corollary of Theorem \ref{thm:coh-hard} we obtain the following result.

\begin{corollary}\label{thm:conditional-non-coh}
    Let $k\ge 2$ be a natural number and let $f$ be a polynomial in one variable.
    There is a $\PDL$-formula $\varphi$ such that for every team $T$ witnessing that $\varphi$ is not $k$-coherent, we have that $|T|> f(|\varphi|)$.
    %any team $T$ witnessing that $\varphi$ is not $k$-coherent has size  strictly greater than $f(|\varphi|)$.
    %,unless $\Pi_3^{\mathsf P} =\PiE{2}$.\label{thm:conditional-non-coh}
\end{corollary}
\begin{proof}
%, assume that
%for every $\PDL$-formula $\varphi$, if there is a team $T$ witnessing that $\varphi$ is not $k$-coherent, then there is such  a team $T$ with size $|T|\leq f(|\varphi|)$.
    Towards a contradiction, assume $f(n)=n^r$ is such that for every $\PDL$-formula $\varphi$, if $\varphi$ is not $k$-coherent, then there exists a team $T$
    of size $|T|\leq |\varphi|^r$ such that $T \not\models \varphi$, but for all subteams $S\subseteq T$ such that $|S|= k$ it holds that $S\models \varphi$. Therefore, the following statements are equivalent:
    \begin{itemize}
        \item $\varphi$ is $k$-coherent,
        \item for all teams $T$ with $|T|\leq |\varphi|^r$, we have
    that {if every } $S \subseteq T$
    { with } $|S|\leq k$
    { satisfies } $S \models \varphi$,
    { then }$T \models \varphi$.
    \end{itemize}
    %\miika{Above, 2 should be k?} % probably yes OK
    Since the model checking problem for $\PDL$ is in $\NP$,
    this implies \kCOHERENCE is in $\Pi^{\mathsf P}_2$.
    By Theorem~\ref{thm:coh-hard}, it then follows that  $\PiE{2}$ collapses to $\Pi^{\mathsf P}_2$; this contradicts the fact that
    %alternating time hierarchy theorem which asserts that
    $\Pi^{\mathsf P}_2$ is included properly in $\PiE{2}$ (see Baker~et~al.~\cite{doi:10.1137/0204037} and Mocas~\cite{DBLP:journals/tcs/Mocas96}).
    %$f(n)=n^r$ satisfies the claim for all formulas $\varphi$. Then it is easy to see that the argument  of Theorem \ref{thm:coh-mem} gives a $\Pi_3^{\mathsf P}$ upper bound for $k$-coherence yielding the collapse due to the lower bound of Theorem~\ref{thm:coh-hard}.
    %\miika{Find a reference for the above statement, if exists.}
\end{proof}
Another consequence of Theorem \ref{thm:coh-hard} is a lower bound for \kCOHERENCE of quantifier-free dependence logic formulas.
\begin{corollary}
    For every fixed $k\geq 2$, the \kCOHERENCE problem for $\DLQF$ is $\PiE{2}$-hard via  polynomial-time bounded-truth-table reductions.
\end{corollary}
\begin{proof}
We construct a polynomial-time many-one reduction from the \kCOHERENCE problem of $\PDL$ to that of $\DLQF$.
Consider two constant symbols $0$ and $1$.
From a $\PDL$-formula $\varphi(p_1, \dots ,p_n)$, construct a $\DLQF$-formula
%\[
%(\varphi^* \land \lnot 0 = 1 )\lor \lnot p_1 = 0 \lor \lnot p_1 = 1 \lor \dots \lor \lnot %p_n = 0 \lor  \lnot p_n = 1,
%\]
\[
\psi_\varphi\coloneqq \varphi^* \land (\bigwedge_{i=1}^n x_{p_i}=1 \lor x_{p_i}=0) \land\lnot\, 0=1,
\]
where $ \varphi^*$ is obtained from $\varphi$ by replacing every occurrence of a propositional variable $p_i$ with with  $x_{p_i}=1$, and
every occurrence of a negated propositional variable $\lnot p_i$ with  $x_{p_i}=0$.

Let $\calA$ be a first-order model and $T$ a team. Clearly, if $T$ is not a propositional team, then $\calA \not\models_T \psi_\varphi$. If $T$ is a propositional team, then $\calA \models_T \psi_\varphi$ if and only if $\calA\upharpoonright \{0,1\} \models_T \psi_\varphi$ if and only if $T\models \varphi$. Therefore, we have that $\phi$ is $k$-coherent if and only if $\psi$ is $k$-coherent.
\end{proof}

\section{Outlook}
In this article, we have studied structural and algorithmic aspects of coherence of dependence logic formulas.
The following questions warrant further study:
\begin{itemize}
    \item Is the \kCOHERENCE problem of quantifier-free dependence logic formulas decidable?
    \item Is it possible to replace the truth-table reduction used in the proof of Theorem~\ref{thm:coh-hard} by a Karp-reduction?
    % \item What are the structural and algorithmic aspects of coherence for other team logics?
\end{itemize}
Beyond these questions it is worth exploring the structural and algorithmic aspects of coherence for other team logics, such as the extension of dependence logic with intuitionistic implication~\cite{DBLP:journals/sLogica/Yang13}.

\bibliography{references}

\end{document}